\documentclass[10pt,letterpaper]{article}
\usepackage{amsfonts}
\usepackage{amsmath}
\usepackage{amsthm}
\usepackage{amssymb}
\usepackage{framed}
\usepackage[noend]{algpseudocode}
\usepackage{xcolor}
\usepackage{multirow}
\usepackage{float}
\usepackage{centernot}
\usepackage{graphicx}
\usepackage{subcaption}
\usepackage{fullpage}

\makeatletter
\newcommand{\xMapsto}[2][]{\ext@arrow 0599{\Mapstofill@}{#1}{#2}}
\def\Mapstofill@{\arrowfill@{\Mapstochar\Relbar}\Relbar\Rightarrow}
\makeatother

\newcommand{\vir}[1]{``{#1''}}

\newtheoremstyle{mystyle}{}{}{\itshape}{}{\bfseries}{.}{6 pt}{\thmname{#1}\thmnumber{ #2}\thmnote{ {\bfseries(#3)}}}
\theoremstyle{mystyle} 
\newtheorem{theorem}{Theorem}[section]
\newtheorem{corollary}{Corollary}

\newtheorem{lemma}{Lemma}

\newtheorem{definition}{Definition}

\title{On the Power of Weaker Pairwise Interaction:\newline
Fault-Tolerant   Simulation of Population Protocols}
\author{G. Di Luna\footnote{ University of Ottawa, \texttt{ \{gdiluna,flocchin\}@site.uottawa.ca, viglietta@gmail.com} }$^*$, P. Flocchini$^*$, T. Izumi\footnote{Nagoya Institute of Technology, Gokiso-cho, Showa-ku, Nagoya, Aichi, 466-8555, Japan \texttt{t-izumi@nitech.ac.jp}} , T. Izumi\footnote{ Ritsumeikan University \texttt{izumi-t@fc.ritsumei.ac.jp}} , N. Santoro\footnote{School of Computer Science, Carleton University,  \texttt{santoro@scs.carleton.ca}}, G. Viglietta$^*$ }

\date{}

\begin{document}

\maketitle


\begin{abstract}
In this paper we investigate the computational power of Population Protocols (PP) under some unreliable and/or weaker interaction models. More precisely, we focus on two features related to the power of interactions: omission failures and one-way communications. An omission failure, a notion that this paper introduces for the first time in the context of PP, is the loss by one or both parties of the information transmitted in an interaction. The failure may or may not be detected by either party. On the other hand, in one-way models, communication happens only in one direction: only one of the two agents can change its state depending on both agents' states, and the other agent may or may not be aware of the interaction. These notions can be combined, obtaining one-way protocols with (possibly detectable) omission failures. 

A general question is what additional power is necessary and sufficient to completely overcome the weakness of one-way protocols and enable them to simulate two-way protocols, with and without omission failures. As a basic feature, a simulator needs to implement an atomic communication of states between two agents; this task is further complicated by the anonymity of the agents, their lack of knowledge of the system, and the limited amount of memory that they may have. 

We provide the first answers to these questions by presenting and analyzing several simulators, i.e., wrapper protocols converting any protocol for the standard 
two-way model into one running on a weaker one. 
\end{abstract}
\newpage


\section{Introduction}

\subsection{Framework}

The {\em Population Protocol} ({\sf PP})  model~\cite{first} is a  mathematical model that describes 
 systems of simple mobile computational entities, called \emph{agents}.
Two agents  can interact (i.e., exchange information) only when their movement brings them  into communication range of each other; however, 
  the  movements of the agents, and thus the occurrences of their interactions,  are completely unpredictable, a condition called ``passive mobility''. Such would be, for example, the case of a flock of birds, each provided with a sensor; the resulting passively mobile sensor network
can then be used for monitoring the activities of the flock and  for individual intervention, such as a sensor  inoculating  the bird  with a drug,
should a certain condition be detected.

In {\sf PP}, when an interaction occurs, the states of the two agents involved  change according to a set of  deterministic rules, or ``protocol''.   The execution of the protocol,  through the interactions  originating from the movements of the entities, 
generates a non-deterministic sequence  of   changes in the states of the entities themselves and, thus in  the global state of the system. 

In an interaction, communication is generally assumed to be bidirectional or {\em two-way}: each agent of a pair  receives  the state of the other agent and applies the protocol's transition function to update its own state,  based on the received information and its current state. From an engineering standpoint, this round-trip communication between two interacting agents may be difficult to implement. Moreover, the standard {\sf PP} model is not resilient to faults.

In this paper we investigate the computational power of {\sf PP} under some unreliable and/or weaker interaction models. More precisely, we focus on two features related to the power of interactions: \emph{omission failures} and \emph{one-way} communications. An omission failure, a notion that this paper introduces for the first time in the context of {\sf PP}, is the loss by one or both parties of the information transmitted in an interaction. The failure may or may not be \emph{detected} by either party. On the other hand, in one-way models (originally introduced in~\cite{oneway}),   communication occurs only in one direction: only one of the two agents can change its state depending on both states, and the other agent may or may not be aware of the interaction. These notions can be combined, obtaining one-way protocols with (possibly detectable) omission failures.

A general question is what additional power  is necessary and sufficient to fill the gap between the standard two-way model
and the weaker models stated above. In this paper we start to address this question, using  as a main investigation tool
 the  concept of a {\em simulator}:  a wrapper protocol converting any protocol for the standard 
two-way model into  one running on some weaker model. A simulator provides an interface between the simulated protocol 
and the physical communication layer, giving the system the illusion of being in a two-way environment. 
As a basic feature, a simulator has to implement an atomic communication of states between two agents,
always guaranteeing both safety and liveness of any problem specification;
this task is further complicated by the anonymity of the agents, their lack of knowledge of the system, and the limited amount of memory that they may have.

\subsection{Main Contributions}

We consider 
  the computationally distinct
models that arise  
from the introduction of omission faults and/or 
one-way behavior in two-way  protocols (see  Figure \ref{figure:modcompress}).
In particular, {\sf TW} refers to two-way protocols without omissions; {\sf IT} and {\sf IO} refer to the one-way models {\em Immediate Transmission} and {\em Immediate Observation}, introduced in \cite{oneway};  the {\sf T$_i$}'s  and {\sf I$_i$}'s refer to the distinct two-way and one-way model with omissions,  respectively.
  
  We consider 
 two main types of  omission adversaries: a ``malignant'' one, called {\bf UO}, who can insert omissions at any point in the execution, and a ``benign'' one, called $\diamond${\bf NO}, who must eventually stop inserting omissions. Interactions are otherwise ``globally fair''. Interestingly, all our main simulators work even under the malignant {\bf UO} adversary, while all our main impossibility results hold even under the benign $\diamond${\bf NO} adversary.
  
We start by analyzing the negative impact that omissions have on computability.
 We show that, in the absence of additional assumptions, the
simulation of {\sf TW} protocols in the presence of omissions is impossible
even if the agents have infinite memory (Theorem~\ref{th:impoinfinite}).
Among other results, we also show that,
  in the two weak omission models {\sf I$_1$} and {\sf I$_2$}, simulation is impossible even under an extremely limited omission adversary, called $\diamond${\bf NO}$\mathbf{_1}$, which can only insert at most one omission in the entire execution.

  On the other hand we prove that,  in the weakest one-way model,  {\sf IO}, simulation is possible if the agents have unique IDs or the total number of agents, $n$, is known (Theorems~\ref{th2:simid1} and~\ref{th2:simN}).

  In the two strong omission models  {\sf I$_3$} and {\sf I$_4$}, simulation is possible when an upper bound on the number of omissions is known (Theorem~\ref{th2:simO}). This result in turn implies that,
in the non-omissive {\sf IT} model,
{\sf TW} simulation is possible
 with a memory overhead of ${\Theta}(\log n)$ bits 
 for each state of the simulated protocol (Corollary~\ref{obs:1}). 
  In light of the fact  that 
with constant memory, in absence of additional capabilities,  {\sf IT} protocols are strictly 
less powerful than two-way protocols \cite{oneway}, our results show that this computational gap can
be overcome by using
additional memory.

Our main results  are summarized in Figure~\ref{id:algorwq2}, where green 
blobs represent possibilities, and red 
blobs impossibilities. As a consequence of these results, we have a complete characterization
 of the feasibility of   simulation  with respect to infinite memory  and knowledge of the size of the system.

\begin{figure}
\makebox[\textwidth][c]{\hspace{-0.5cm}\includegraphics[scale=0.6705]{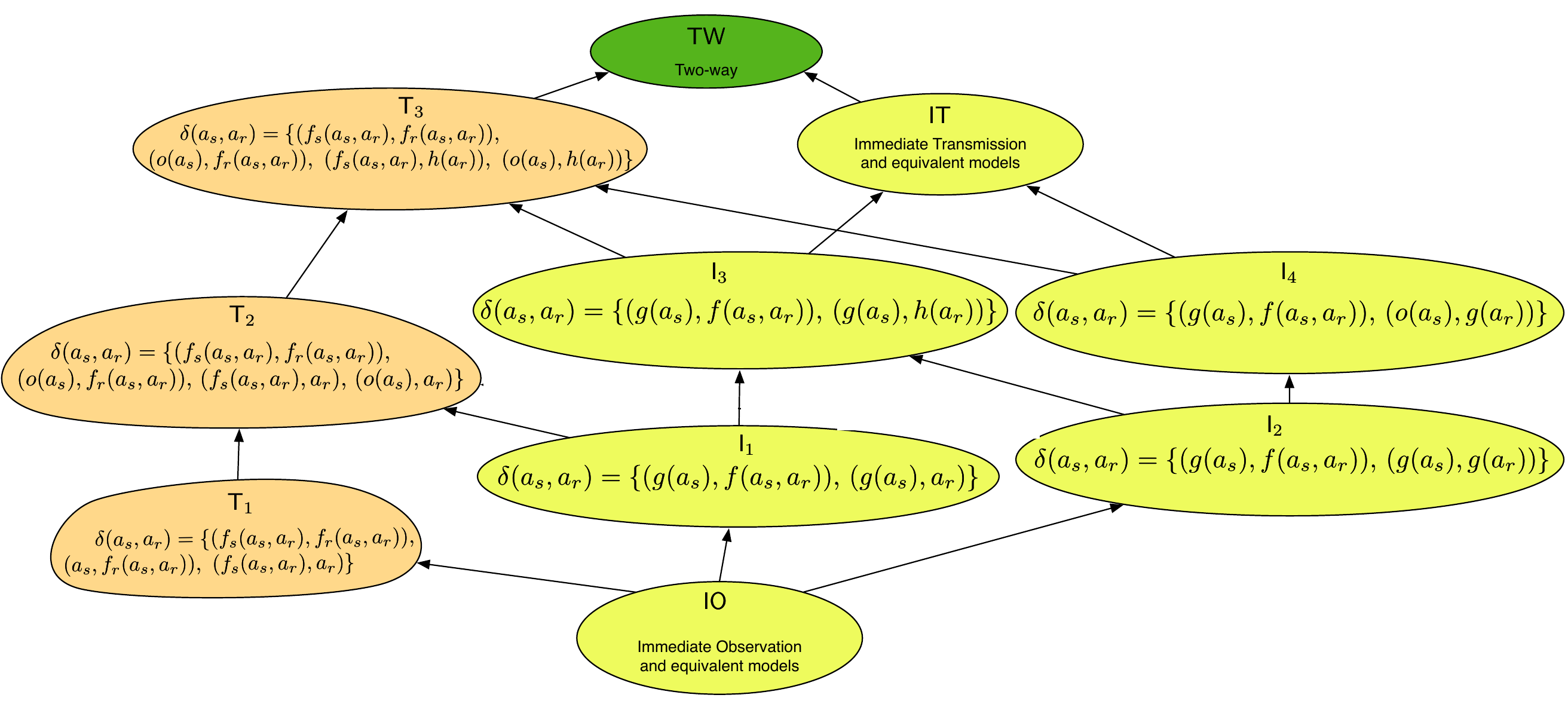}}
\caption{\footnotesize  Computational relationships between models. An arrow between two blobs indicates that the class of solvable problems in the source blob is included in that of the destination blob. The models on the left,  {\sf T$_1$}, {\sf T$_2$}, {\sf T$_3$}, are the two-way models with omissions. The models on the right, {\sf I$_1$}, {\sf I$_2$}, {\sf I$_3$}, {\sf I$_4$},
are the one-way models with omissions.  \label{figure:modcompress}}
\end{figure}

%
%
%
%


\subsection{Related Work}

Since their introduction, there have been  extensive  investigations  on Population Protocols (e.g., see  \cite{computability,first2,infinitepopulation,scheduler,popbook,probabilistic,popchecm,ids}), and the basic assumptions of the original model have been expanded in several directions, typically  to overcome inherent computability  restrictions. 
For example, allowing each agent to have non-constant memory~\cite{dalistar2,dalistar,passivemachine}; assuming the presence of  a leader \cite{fastleader}; allowing 
 a certain amount of information to be stored on the edges~\cite{mediated,mediated3,mediated2} of the interaction graph.

 
The  issue of dependable computations in {\sf PP}, first raised in   \cite{firstfault}, has been considered and studied only with respect to processors' faults, and the basic model has necessarily been expanded.
In \cite{fault} it has been shown how to compute functions tolerating ${\cal O}(1)$ crash-stops and transient failures,
assuming that  the number of failures is bounded and known.
In \cite{majbiz}  the specific majority problem under ${\cal O}(\sqrt{n})$ Byzantine failures,  
 assuming a fair probabilistic scheduler, has been studied. 
In \cite{bizfault} unique IDs are assumed, and it is shown how to compute functions tolerating a bounded number of Byzantine faults, under the assumption that Byzantine agents cannot forge IDs. 
Self-stabilizing solutions have been devised for specific problems such as 
 leader election (assuming knowledge of the system's size  and a non-constant number of states~\cite{izumile},
                       or assuming a leader detection oracle~\cite{leaderelection}) and
 counting  (assuming the presence of  a   leader   \cite{spaceoptcounting}). 
 Moreover,  in \cite{Beauquier20114247}  a self-stabilizing transformer for general protocols has been studied in a slightly different model and under the assumption of unbounded memory and a leader. 

Finally, to the best of our knowledge,  the one-way model without omissions, has been studied   only in \cite{oneway}, where it is shown that {\sf IT} and {\sf IO}, when equipped with constant memory, can compute a set of functions that is strictly included in that of {\sf TW}.   Combined with our results in Figure~\ref{id:algorwq2}, this implies that, without using extra resources (e.g., infinite memory, leader, etc.), simulations are impossible in all the one-way and omissive models.


%

\section{Models and Terminology}
\subsection{Population Protocols}
We consider a system consisting of a set $A = \{a_{1},\ldots,a_{n}\}$ of mobile agents. The mobility is passive, in the sense that it is decided by an external entity. When two agents meet, they interact with each other and perform some local computation. We always assume that interactions are instantaneous. Each interaction
is asymmetric, that is, an interaction between $a_s$ and $a_r$ is indicated by the ordered pair $i=(a_s, a_r)$, where $a_s$ and $a_r$ are called \emph{starter} and 
\emph{reactor}, respectively. A protocol $\mathcal{P}$ is defined by the following 
three elements: a set of local states $Q_{\mathcal{P}}$, a set of initial states $Q'_{{\cal P}} \subseteq {Q}_{\mathcal{P}}$, 
and a transition function $\delta_{\cal P}\colon Q_{{\cal P}} \times Q_{{\cal P}} \rightarrow 
Q_{{\cal P}} \times Q_{{\cal P}}$. The function $\delta_{\cal P}$ 
defines the states of the two interacting agents at the end of their local 
computation. With a small abuse of notation, and when no ambiguity arises, we will use the same literal (e.g., 
$a_i$) to indicate both an agent and its internal 
state. Since the static structure of the system is uniquely determined by 
$\mathcal{P}$ and $n$, we refer to it as the \emph{system} $(\mathcal{P}, n)$. 
A \emph{configuration} $C$ of a system $(\mathcal{P},n)$ is the $n$-tuple of local 
states in $Q_{\mathcal{P}}$ (i.e., $C \in Q_{\mathcal{P}}^n$).

Given an $k$-tuple 
$t = (x_0,x_1,\ldots,x_{k-1})$ we denote the element $x_j$ by $t[j]$. 
\medskip

 \noindent {\bf Initial Knowledge.} 
 To empower the agents, we sometimes assume that each agent has some additional knowledge, such as unique IDs and/or knowledge of $n$.  We model this information by encoding it as a set of initial states of the agents (i.e., in $Q'_{{\cal P}}$). 
\medskip

\noindent {\bf Executions and Fairness.} 
Whenever an interaction $i=(a_j,a_k)$ turns a configuration of the form
$C = (a_1,\ldots,a_j,\ldots,a_k,\ldots,a_n)$ into one of the form $$C' = (a_1,\ldots,\delta(a_j,a_k)[0],
\ldots,\delta(a_j,a_k)[1],\ldots,a_n),$$ we use the notation
$C \xrightarrow{i} C'$. 
A \emph{run} of ${\cal P}$ is an infinite sequence of interactions 
$I = (i_0,i_1,\ldots)$. Given an initial 
configuration $C_{0}\in Q'^n_{{\cal P}}$, each run $I$ induces an infinite sequence of configurations, 
$\Gamma_{I}(C_0) =(C_{0}, C_{1},\dots)$ such that $C_{j} \xrightarrow{i_{j}} C_{j+1}$ for every $j \geq 0$, which is called an \emph{execution}
of $\mathcal{P}$.

We say that a set of configurations $\mathcal C \subseteq Q^n_{\mathcal{P}}$ is \emph{closed} if, for every $C\in\mathcal C$, and for every configuration $\widehat C$ obtained by permuting the states of the agents of $C$, also $\widehat C\in\mathcal C$.

An execution $\Gamma$ is \emph{globally fair} (\textsf{GF}) if it satisfies the following condition: for every two (possibly infinite) closed sets of configurations $\mathcal C, \mathcal C' \subseteq Q^n_{\mathcal{P}}$ such that for every $C\in \mathcal C$ there exists an interaction $i$ and some $C'\in \mathcal C'$ such that $C \xrightarrow{i} C'$, if infinitely many configurations of $\Gamma$ belong to $\mathcal C$, then infinitely many configurations of $\Gamma$ belong to $\mathcal C'$ (although not necessarily appearing in $\Gamma$ as immediate successors of configurations of $\mathcal C$).

Note that our definition of global fairness extends the standard one, which only deals with single configurations, as opposed to sets (see~\cite{first2}). The two definitions are equivalent when applied to protocols that use only finitely many states, but our extension also works with infinitely many states, while the standard one is ineffective.
\subsection{Interaction Models}

In this paper we consider three main models of interactions: the standard {\em Two-Way} one, and two one-way models presented in \cite{oneway}, i.e., the
\emph{Immediate Transmission} model and the \emph{Immediate Observation} model.

\smallskip

\noindent {\bf Two-Way Model (}\textsf{TW}{\bf ).} In this model, any protocol 
$\mathcal{P}$ must have a state transition function consisting of
two functions $f_s \colon Q_{\mathcal{P}} \times Q_{\mathcal{P}} \to Q_{\mathcal{P}}$ and 
$f_r\colon Q_{\mathcal{P}} \times Q_{\mathcal{P}} \to Q_{\mathcal{P}}$ satisfying
$\delta_{\mathcal{P}}(a_s,a_r) = (f_s(a_s,a_r), f_r(a_s,a_r))$.
\smallskip

\noindent {\bf Immediate Transmission Model (}\textsf{IT}{\bf ).} Any protocol 
$\mathcal{P}$ must have a state transition function consisting of 
two functions $g \colon Q_{\mathcal{P}} \to Q_{\mathcal{P}}$ and 
$f\colon Q_{\mathcal{P}} \times Q_{\mathcal{P}} \to Q_{\mathcal{P}}$ satisfying
$\delta_{\mathcal{P}}(a_s,a_r) = (g(a_s), f(a_s,a_r))$ for any $a_s, a_r \in Q_{\mathcal{P}}$.
\smallskip

\noindent {\bf Immediate Observation Model (}\textsf{IO}{\bf ).} Any protocol 
$\mathcal{P}$ must have a state transition function of the form 
$\delta_{\mathcal{P}}(a_s, a_r) = (a_s, f(a_s,a_r))$.
\smallskip

Note that, in the {\sf IT} model, the starter explicitly detects the interaction, as it applies function $g$ to its own state. In other terms, even if the starter cannot read the state of the reactor, it can still detect its ``proximity''. In the {\sf IO} model, on the other hand, there is no such detection of an interaction (or proximity) by the starter.

%
%


\subsection{Omissive Models}
An omission is a fault affecting a single interaction. In an omissive interaction an agent does not receive any information about the state of its counterpart. Omissions are introduced by an adversarial entity. We consider:
\begin{definition}[Unfair Omissive (UO) Adversary] The {\bf UO} adversary takes a run $I$ and outputs a new sequence $I'$, which is obtained by inserting a (possibly empty) finite sequence of omissive interactions between each pair of consecutive interactions of $I$. 
\end{definition}
\begin{definition}[Eventually Non-Omissive ($\diamond${\bf{NO}}/$\diamond${\bf NO}$\mathbf{_1}$) Adversary] The $\diamond${\bf NO} adversary takes a run $I$ and outputs a new sequence $I'$, which is obtained by inserting any finite sequence of omissive interactions between finitely many pairs of consecutive interactions of $I$. The $\diamond${\bf NO}$\mathbf{_1}$ adversary is even weaker, and can only output interaction sequences with at most one omission.
\end{definition}
\smallskip
If we incorporate omissions in our runs, then transition functions become more general relations.

\smallskip
\noindent {\bf {\sf TW} Omissive Model.}
In the two-way omissive model, we have the  transition relation $$\delta(a_s,a_r)=\{(f_s(a_s,a_r), f_r(a_s,a_r)), (o(a_s),f_r(a_s,a_r)),\ (f_s(a_s,a_r),h(a_r)),\ (o(a_s),h(a_r))\}$$
(model {\sf T$_3$}).
The first pair is the outcome of an interaction when no omission is present; the other three pairs represent all possible outcomes when there is an omission: respectively, an omission on the starter's side, on the reactor's side, and on both sides. The functions $o$ and $h$ represent the detection capabilities of each agent: in {\sf TW}, if one of these is the identity, then omissions are \emph{undetectable} on the respective side.

\smallskip
\noindent {\bf One-Way Omissive Models.}
In the case of one-way interactions, we have the  transition relation $\delta(a_s,a_r)=\{(g(a_s),f(a_s,a_r)),\ (o(a_s),h(a_r))\}.$ The first pair is the outcome of an interaction when no omission is present, and the second pair when there is an omission. (Note that the {\sf IO} model corresponds to the case in which $g$ is the identity function.) Once again, omissions are undetectable starter-side (respectively, reactor-side) if $o$ (respectively, $h$) is the identity function.

\smallskip
\noindent {\bf Hierarchy of Models.}
 The previous models can be weakened by removing the omission detection, either on the starter's side, or on the reactor's side. After identifying all possible combinations of omissions and detections, and pruning out the equivalent ones, the significant models and their relationships have been reported in Figure~\ref{figure:modcompress}. 
For {\sf TW} omissive models, in {\sf T$_2$} we have the models where there is no detection of omission either on the starter's or the reactor's side. Since these two models are symmetric, only the one without reactor-side detection is reported, i.e., function $h$ is forced to be the identity. In {\sf T$_1$} we have the weaker model where no detection is available, i.e., both $o$ and $h$ are the identity. In one-way models, function $g$ is applied when an agent detects the \emph{proximity} of another agent. However, this does not imply the detection of an omission: in {\sf I$_2$}, no agent detects an omission, but both detect the proximity of the other agent. 
 
Each arrow in Figure~\ref{figure:modcompress} indicates either  the obvious inclusion, that is, the transition relation of the source is a special case of the transition relation of the destination, or that the adversary can force the inclusion by avoiding omissions (this is the case with {\sf T$_3$} and {\sf TW}, for instance). Thus, arrows also indicate inclusions of the sets of problems that are solvable in the various models.

\subsection{Simulation of Two-Way Protocols}
\label{s:simulation}
In this section we define the \emph{two-way protocol simulator} (or ``simulator'' for short) and other 
related concepts. Given a two-way protocol $\mathcal P$, consider a protocol
$\mathcal{S}(\mathcal{P})$, whose set of local states is $ Q_{\mathcal{P}} \times Q_{\mathcal{S}}$, where $ Q_{\mathcal{P}}$ is the set of local states of $\mathcal P$ (the ``simulated states''), and $ Q_{\mathcal{S}}$ is additional memory space used in the simulation. Let
$\pi_{\mathcal{P}}\colon  Q_{\mathcal{P}} \times  Q_{\mathcal{S}} \to Q_{\mathcal{P}}$ be the projection function onto the set of local states of $\mathcal P$. By extension, if $C$ is a configuration of $\mathcal{S}(\mathcal{P})$, we write $\pi_\mathcal P(C)$ to indicate the configuration of $\mathcal P$ consisting of the projections of the states of the agents of $C$.

Given an execution $\Gamma_I(C_0)$ of $\mathcal{S}(\mathcal{P})$, where $I = (i_0, i_1, \dots)$, we say that $E(\Gamma) = (e_0, e_1, \dots)$ is a sequence of \emph{events} for $\Gamma$ if it is a weakly increasing sequence of indices of interactions of $I$, such that no three indices are the same, and containing at least the indices of the interactions that determine the update of the simulated state of some agent in the execution $\Gamma$ (if an interaction updates the simulated states of two agents, then its index must appear twice in $E(\Gamma)$). So, with each event $e_j$ in $E(\Gamma)$, we can associate a unique agent involved in the interaction $i_{e_j}$; preferably, this agent is one that effectively changes simulated state as a result of $i_{e_j}$. We also allow extra events in $E(\Gamma)$, associated with agents that do not change simulated state, because we want to take into account simulations of two-way protocols that occasionally leave the state of an agent unchanged.

If $\Gamma_I(C_0)=(C_0, C_1,\dots)$, we let $C^-_j = C_{e_j}$ and $C^+_j = C_{e_j+1}$. In other words, $C^-_j$ and $C^+_j$ are the configurations before and after the $j$-th update of the simulated state, respectively.
\begin{definition}[Perfect matching of events]\label{def:matching}
Given an execution of $\Gamma_I(C_0)$ of a run $I$ and a sequence of events $E(\Gamma)$, a \emph{perfect matching} 
$M(E)$ is a partition of $\mathbb N$ into ordered pairs (viewed as indices of events of $E(\Gamma)$) such that, if
$({e_j}, {e_k}) \in M(E)$, where ${e_j}$ is associated with agent $a_x$ and ${e_k}$ with agent $a_y$, then $x \neq y$ and
$$\delta_{\mathcal{P}}(\pi_{\mathcal{P}}(C^-_j[x]), \pi_{\mathcal{P}}(C^-_k[y])) = (\pi_{\mathcal{P}}(C^+_j[x]), \pi_{\mathcal{P}}(C^+_k[y])).$$
\end{definition}
Intuitively, a pair $(e_j, {e_k})$ in a perfect matching is the pair of
events representing the two state changes given by a two-way interaction of agents under the simulated protocol $\mathcal P$. The events ${e_j}$ and ${e_k}$ correspond to 
the updates of the simulated states of the starter and the reactor, respectively.  
A  matching $M(E)$ induces a \emph{derived run} $D$ of $\mathcal P$ as follows. Sort the pairs $({e_j}, {e_k})$ of $M(E)$ by increasing $\min\{e_j,e_k\}$, and let $M'$ be the sorted sequence. Now, if $({e'_j}, {e'_k})$ is the $m$-th element of $M'$, agent $a_x$ is associated with event ${e'_j}$ and agent $a_y$ is associated with event ${e'_k}$, then the $m$-th element of $D$ is $(x,y)$. Now, the \emph{derived execution} induced by $M(E)$ is simply the execution of $\mathcal P$ induced by $D$, i.e., $\Gamma_{D}(\pi_{\mathcal P}(C_0))$. 
\begin{definition}[Simulation]\label{def:simulator}
A protocol $\mathcal{S}(\mathcal{P})$ \emph{simulates}  $\mathcal{P}$ if, for any initial configuration $C_0$ of $n$ agents of $\mathcal{S}(\mathcal{P})$, and any run $I$ whose execution $\Gamma_I(C_0)$ satisfies the {\sf GF} condition, there exists a sequence of events $E(\Gamma)$ with a  perfect matching $M(E)$ whose derived execution is an execution of $n$ agents of $\mathcal P$ starting from the initial configuration $\pi_{\mathcal P}(C_0)$ and satisfying the {\sf GF} condition. We further require that, for each initial configuration $C_0$, every finite initial sequence of interactions of $\mathcal{S}(\mathcal{P})$ (possibly with omissions) can be extended to an infinite one $I$, having no additional omissions, whose execution $\Gamma_I(C_0)$ satisfies the {\sf GF} condition.
\end{definition}
The last clause of the definition has been added because, with infinite-memory protocols, the existence of {\sf GF} executions cannot be taken for granted.


\section{Impossibilities for  Simulation in Presence of Omissions \label{imp1}}

 In this section,   we derive  several impossibility results    in the presence of omissions. All our impossibility  proofs rely on the existence of a two-way protocol that cannot be  simulated. 
\begin{definition}[Pairing Problem]
A set of agents $A$ is given,  partitioned into  \emph{consumer} agents $A_c $,  starting in state   $c$, 
and \emph{producer} agents  $A_p$, starting in state $p$. 
We say that a protocol $\mathcal P$ solves the \emph{Pairing Problem (\textsf{Pair})} if it enforces the following properties:
\begin{itemize}
\item {\bf Irrevocability.} $\mathcal P$ has a state $cs$ that only agents in state $c$ can get; once an agent has state $cs$, its state cannot change any more.
\item {\bf Safety.} At any time, the number of agents in  state $cs$ is at most $|A_p|$. 
\item {\bf Liveness.} In all {\sf GF} executions of $\mathcal P$, eventually the number of agents in  state $cs$ is stably equal to $\min\{|A_c|,|A_p|\}$.
\end{itemize}

\end{definition}  

It is easy to see that      \textsf{Pair} can be solved by the  simple protocol below in the standard two-way model.

\begin{framed}
\hspace{-0.675cm}
\footnotesize
{\sf  Pairing  Protocol} $\mathcal{P}_{IP}$.
$Q_{\mathcal{P}_{IP}} = \{cs,c,p,\bot\}$.  The only non-trivial transition rules are $(c,p) \mapsto (cs,\bot)$ and $(p,c) \mapsto (\bot,cs)$.
\end{framed}

%
Let us now define a property on the behavior of a generic simulator ${\cal S(\mathcal P)}$ over a sequence of interactions $I$. We will later show how this property is related to the omission resilience of ${\cal S(\mathcal P)}$.  
\begin{definition}[Transition Time ({\sf TT})] Given a {\sf TW} protocol ${\cal P}$, a simulator ${\cal S(\mathcal P)}$, and an execution $\Gamma=(C_0,C_1,\ldots)$ of ${\cal S(\mathcal P)}$ on a system of two agents, the {\em Transition Time ({\sf TT})} of the triplet $({\cal S},{\cal P},\Gamma)$ is the smallest $t$ such that $\pi_\mathcal P(C_t[0])=\delta_\mathcal P(\pi_\mathcal P(C_0[0]),\pi_\mathcal P(C_0[1]))[0]$ and $\pi_\mathcal P(C_t[1])=\delta_\mathcal P(\pi_\mathcal P(C_0[0]),\pi_\mathcal P(C_0[1]))[1]$ (or $\infty$, if no such $t$ exists).
\end{definition}
Let $O(I)$ be the number of omissions in a sequence of interactions $I$.
\begin{definition}[Fastest Transition Time ({\sf FTT})] Given a {\sf TW} protocol ${\cal P}$, a simulator ${\cal S(\mathcal P)}$, and a configuration $C_0$ for a system of two agents of ${\cal S(\mathcal P)}$, the {\em Fastest Transition Time ({\sf FTT})} of the triplet $({\cal S},{\cal P},C_0)$ is the smallest {\sf TT} of all the triplets of the form $({\cal S},{\cal P},\Gamma_I)$, where $I$ ranges over all runs with $O(I)=0$ and $\Gamma_I[0]=C_0$.
\end{definition}

Intuitively, {\sf FTT} is the minimum number of (non-omissive) interactions needed by a specific simulator ${\cal S}$ to simulate one step of protocol ${\cal P}$ in a system of two agents. Thus it can be seen as the \vir{maximum speed} of a simulator. We will show in the following that such a metric is intrinsically related with the omission resilience of ${\cal S}$.

\subsection{Impossibilities in Spite of Infinite Memory \label{imp2:infmemory}}

In this section we show that   simulations of {\sf TW} models are impossible when omissions are present, even if the system is endowed with infinite memory. We start presenting a key indistinguishability argument.

\begin{lemma}\label{lemma:indstinguishability1}
Let ${\cal S(\mathcal P)}$ be a simulator working in the omission model {\sf T$_3$}.
Let $t>0$ be the {\sf FTT} of the triplet $({\cal S},{\cal P},C_0)$, where one agent in $C_0$ has simulated state $q_0$, the other agent has $q_1$ with $q_0\neq q_1$, and $\delta_\mathcal P(q_0,q_1)=(q'_0,q'_1)$ and $\delta_\mathcal P(q_1,q_0)=(q'_1,q'_0)$. Let $A$ be a system of $2t+2$ agents of ${\cal S(\mathcal P)}$, and let $B_0$ be an initial configuration of $A$ in which $t$ agents have simulated state $q_0$ and $t+2$ agents have $q_1$. Then, there exists a sequence of interactions $I^\ast$ of $A$ such that $\Gamma_{I^\ast}(B_0)$ is {\sf GF} and $O(I^\ast)=t$, with a sequence of events $E(\Gamma_{I^\ast}(B_0))$ in which at least $t+1$ events represent a transition of some agent from simulated state $q_1$ to $q'_1$.
\end{lemma}
\begin{proof}
Intuitively, we construct a system with $t$ pairs of agents which, thanks to omissive interactions, we ``fool'' into believing that they are operating in a system of only two agents, until one agent per pair transitions from simulated state $q_1$ to $q'_1$. Then we have an extra agent that interacts once with one member of each of the $t$ pairs, also ``believing'' that the system consists of only two agents, which finally transitions from simulated state $q_1$ to $q'_1$. One last auxiliary agent serves as a ``generator'' of omissive interactions.

Let $I$ be any run of a system of two agents achieving {\sf FTT} for $({\cal S},{\cal P},C_0)$; let $d_0$ be the agent whose initial state is $q_0$ and let $d_1$ be the other one. For every $0\leq k<t$, we construct a sequence of interactions $I_k$ for two agents as follows: copy the first $k$ interactions from $I$; append an omissive interaction with the same starter as $I[k]$, and with omission (and detection) on $d_1$'s side; extend the resulting sequence to an infinite one whose execution from $C_0$ satisfies the {\sf GF} condition, without adding extra omissions (such an extension exists by Definition~\ref{def:simulator}, since $\mathcal S$ is a simulator). Note that $I_k$ has exactly one omissive interaction.

Because the execution of $I_k$ is {\sf GF}, the derived execution must also be {\sf GF} by definition of simulator, and in particular it makes the simulated states of the two agents transition according to $\delta_\mathcal P$ infinitely many times. Hence, the agent whose initial simulated state is $q_1$ will eventually transition to $q'_1$, say after the execution of the first $t_k$ interactions of $I_k$. Note that this happens regardless of which agent is the starter of the two-way simulated interaction, because by assumption $\delta_\mathcal P$ is symmetric on $(q_0,q_1)$.

  \begin{figure}[H]

    \begin{subfigure}{0.2\textwidth}
	\vspace{1.2cm}
	\makebox[\textwidth][c]{\includegraphics[scale=0.4]{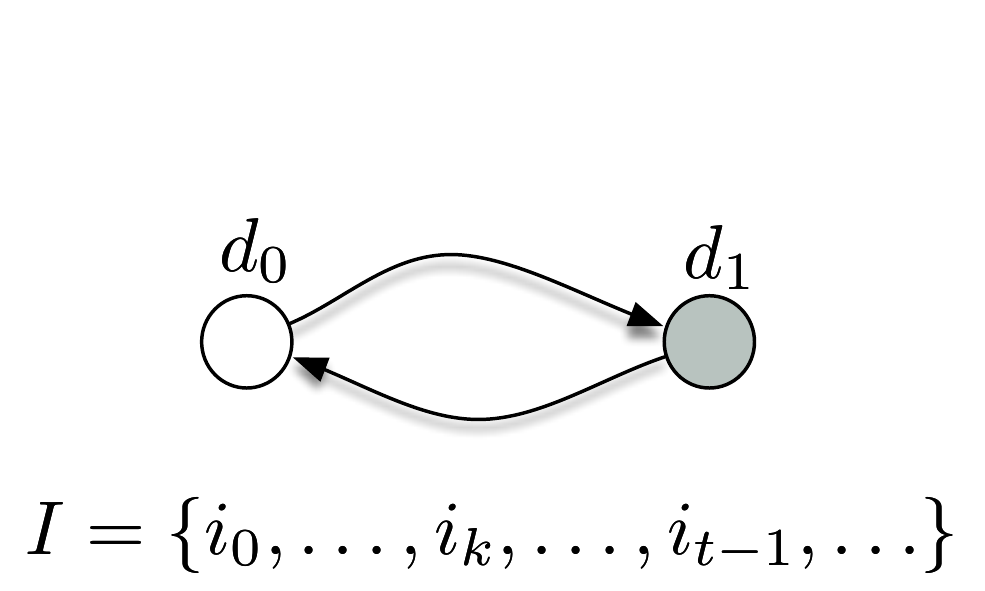}}

        \caption{Run $I$ achieving {\sf FTT} \label{figure:noom}}
    \end{subfigure}
    ~ ~ ~ ~ ~ 
    \begin{subfigure}{0.2\textwidth}
      \vspace{0.75cm}
       \includegraphics[scale=0.4]{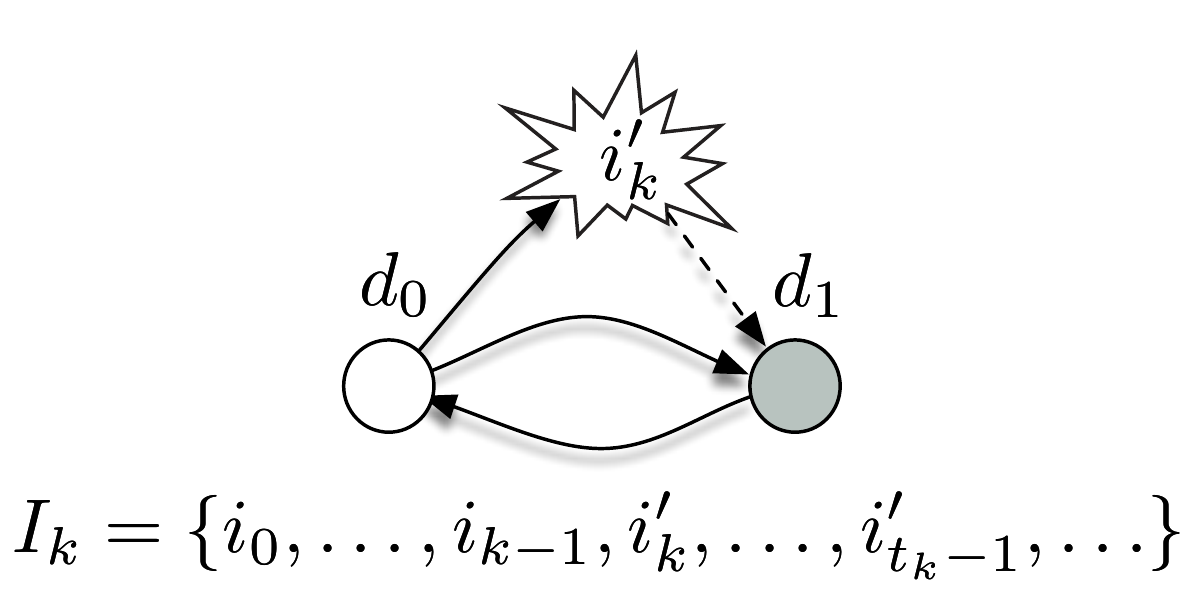}
        \caption{Partial run $I_k$ \label{figure:oneom}}
    \end{subfigure}
   ~ ~ ~ ~ ~ ~ ~ ~ ~
      \begin{subfigure}{0.3\textwidth}
      	\vspace{0.05cm}
    \includegraphics[height=3cm]{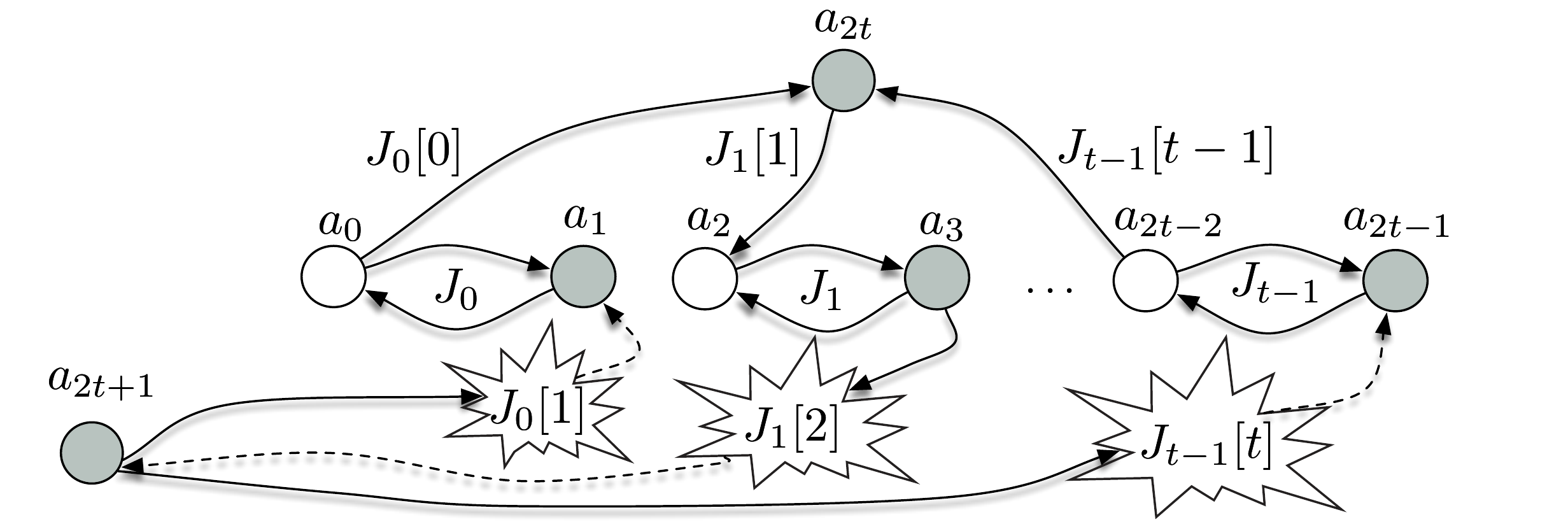}
        \caption{Final run $I^\ast$ \label{figure:impinfmemory2} }
    \end{subfigure}

      \caption{ Construction of the run $I^\ast$  \label{figure:impinfmemory1} }
\end{figure}
%
%
%

Now name the agents of $A$ as in $A=\{a_0,\dots,a_{2t+1}\}$, in such a way that, for all $0\leq k<t$, the agents of the form $a_{2k}$ have simulated state $q_0$ in $B_0$, while all other agents of $A$ have $q_1$. For every $0\leq k<t$, we construct a sequence $J_k$, consisting of $t_k+1$ interactions, involving only agents $a_{2k}$, $a_{2k+1}$, $a_{2t}$, and $a_{2t+1}$. We make $a_{2k}$ and $a_{2k+1}$ interact with each other as in $I_k$, but we ``redirect'' the omissive interaction $I_k[k]$ to $a_{2t}$ and $a_{2t+1}$. Specifically, we replicate the first $k$ interactions of $I_k$ (where $d_0$ becomes $a_{2k}$ and $d_1$ becomes $a_{2k+1}$); then we add an interaction between $a_{2k}$ and $a_{2t}$, where the role of $a_{2k}$ (i.e., starter or reactor) is the same as that of $d_0$ in $I_k[k]$; then we insert an omissive interaction between $a_{2k+1}$ and $a_{2t+1}$, where the role of $a_{2k+1}$ is the same as that of $b_1$ in $I_k[k]$, and the omission (and detection) is on $a_{2k+1}$'s side; finally, we replicate the $t_k-k-1$ interactions of $I_k$ from $I_k[k+1]$ to $I_k[t_{k}-1]$. Observe that $J_k$ contains exactly one omissive interaction, $J_k[k+1]$.

The final sequence $I^\ast$ is now simply the concatenation of all the sequences $J_k$ (where $k$ goes from $0$ to $t-1$, in increasing order), extended to an infinite sequence of interactions whose execution is {\sf GF}, and having exactly $t$ omissions in total (again, this extension exists by Definition~\ref{def:simulator}). 

Let us examine the execution of $I^\ast$ from the initial configuration $B_0$. Each of the $t$ pairs of the form $(a_{2k},a_{2k+1})$, with $0\leq k<t$, has an initial execution that is the same as that of $(d_0,d_1)$ interacting for $k$ turns as in $I_k$. Hence, the execution of $a_{2t}$ is that of $d_1$ interacting as in $I$ for the first $t$ turns. It follows that $a_{2t}$ transitions from simulated state $q_1$ to $q'_1$ by the end of the sub-run $J_{t-1}$. Also, for each pair $(a_{2k},a_{2k+1})$, the execution is as in $I_k$ for the first $t_k$ turns; hence, $a_{2k+1}$ transitions from simulated state $q_1$ to $q'_1$ at the end of the sub-run $J_k$. Thus, in total, we have at least $t+1$ agents that transition from $q_1$ to $q'_1$.
\end{proof}

\begin{theorem}\label{th:impoinfinite}
Given an infinite amount of memory on each agent, it is impossible to  simulate every  {\sf TW} protocol in the {\sf T$_3$} model (hence in all the \emph{omissive} models of Figure~\ref{figure:modcompress}),  even under the  $\diamond${\bf NO} adversary. 
 \end{theorem}
\begin{proof}
We show that the protocol $\mathcal P_{IP}$ for {\sf Pair} cannot be simulated  if any type of omissive interaction is allowed. Assume by contradiction that there is  a simulator $\mathcal S$ for $\mathcal P_{IP}$, i.e.,  $\mathcal S$ tsolves {\sf Pair} under some omissive model. 
Let us now apply Lemma~\ref{lemma:indstinguishability1} to $\mathcal S$ and $\mathcal P_{IP}$, where $q_0$ is the initial state of the providers (hence there are $t$ providers), $q_1$ is the initial state of the consumers (hence there are $t+2$ consumers), and $q'_1$ is the irrevocable state.

Because $\mathcal P_{IP}$ is symmetric with respect to starter and reactor, the hypotheses of Lemma~\ref{lemma:indstinguishability1} are satisfied, and hence there is a sequence of interactions $I^\ast$ whose execution is {\sf GF}, which causes $t+1$ transitions into the critical state. Since the execution is {\sf GF}, the derived execution of $I^\ast$ must be an execution of $\mathcal P_{IP}$, due to Definition~\ref{def:simulator}. In particular, it satisfies the irrevocability property of  {\sf Pair}. Therefore, no agent entering a critical state can ever change it. It follows that, eventually, there are at least $t+1$ agents in the critical state, which contradicts the safety property of  {\sf Pair}.

Since $I^\ast$ contains just finitely many omissive interactions, it can be generated by the $\diamond${\bf NO} adversary.
\end{proof}

Theorem~\ref{th:impoinfinite} uses as counterexample the construction  of Lemma~\ref{lemma:indstinguishability1}, implying that a simulator {\cal S} fails to simulate protocol $\mathcal P_{IP}$ in a run where the number of failures is exactly the {\sf FTT} of $({\cal S}, \mathcal P_{IP},(c,p))$. This is even more interesting if we consider simulators that are unaware of the protocol they are simulating, where by \vir{unaware} we mean that the sequence of simulated two-way interactions is not influenced by the protocol that is being simulated or by the initial configuration (i.e., general-purpose and not ad-hoc simulators).  We have shown that each of these simulators fails as soon as the number of omissions is above some constant threshold, which is independent of the simulated protocol and the initial configuration. Such a threshold is precisely the minimum number of non-omissive interactions needed to simulate a single two-way transition.

For models {\sf I$_{1}$} and {\sf I$_{2}$}, we can strengthen Theorem~\ref{th:impoinfinite}.
\begin{theorem} \label{th:meminfimp1om}
Given an infinite amount of memory on each agent, it is impossible to  simulate every {\sf TW} protocol in the interaction
 models {\sf T$_{1}$}, {\sf I$_{1}$}, and {\sf I$_{2}$}, 
 even under the $\diamond${\bf NO}$\mathbf{_1}$ adversary.
\end{theorem}
\begin{proof} The proof uses a construction analogous to the one used in Lemma~\ref{lemma:indstinguishability1}.
We consider a system $A=\{a_0,\ldots, a_{2t+1}\}$ of $2t+2$ agents, and we build $t$ sequences of interactions $I_k$ between two agents $d_0$ and $d_1$, exactly as in Lemma~\ref{lemma:indstinguishability1}. Recall that the run $I_k$ contains only one omission. Hence, if a simulator is resilient to the $\diamond${\bf NO}$\mathbf{_1}$ adversary, it eventually succeeds in making $d_0$ and $d_1$ simulate a full two-way interaction, say after $t_k$ one-way interactions. Since $t_k$ is well defined, we can go on and construct the sequence $J_k$. However, the $J_k$ that we will use in this proof differs from its counterpart used in Lemma~\ref{lemma:indstinguishability1} by two elements: $J_k[k]$ and $J_k[k+1]$. In particular, our new $J_k$'s will contain no omissions.

If the model is {\sf T$_1$}, we replace the old interactions $J_k[k]$ and $J_k[k+1]$ by a single non-omissive interaction between $a_k$ and $a_{2t}$ (in which $a_k$ is the starter if and only if $d_0$ is the starter in $I[k]$).

Let the model be {\sf I$_1$}. If the interaction $I[k]$ is $(d_{0},d_{1})$, then we replace the old interactions $J_k[k]$ and $J_k[k+1]$ by the single interaction $(a_k,a_{2t})$. Otherwise, if $I[k]=(d_{1},d_0)$, we set $J_k[k]=(a_{2t},a_{2t+1})$ and $J_k[k+1]=(a_{k+1},a_{2t+1})$.

Consider now model {\sf I$_2$}. If the interaction $I[k]$ is $(d_{0},d_{1})$, then we set $J_k[k]=(a_k,a_{2t})$ and $J_k[k+1]=(a_{k+1},a_{2t+1})$. Otherwise, if $I[k]=(d_{1},d_0)$, we replace the old interactions $J_k[k]$ and $J_k[k+1]$ by the three interactions $(a_{2t},a_{2t+1})$, $(a_{k},a_{2t+1})$, and $(a_{k+1},a_{2t+1})$.

Finally, we concatenate the $t$ finite sequences $J_k$ to obtain the new run $I^{*}$, which contains no omissions. Let us now examine the execution of $I^{*}$ from the initial configuration $B_0$ defined in Lemma~\ref{lemma:indstinguishability1}. Once again, each of the $t$ pairs $(a_{2k},a_{2k+1})$, with $0\leq k<t$, has an initial execution that is the same as that of $(d_0,d_1)$ interacting for $k$ turns as in $I_k$. Then, the new interactions that we added in lieu of $J_k[k]$ and $J_k[k+1]$ make $a_{2k}$ and $a_{2k+1}$ change state in the same way as in the omissive interaction $I_k[k]$. But as a side effect, also $a_{2t}$ changes state as it would in a non-omissive interaction with $a_{2k}$. As a consequence, by the end of $I^{*}$, all the agents of the form $a_{2k+1}$ with $0\leq k<t$, as well as $a_{2t}$, have transitioned from simulated state $q_1$ to $q'_1$. Thus, in total, at least $t+1$ agents transition from $q_1$ to $q'_1$.

Now the proof can be completed exactly as in Theorem~\ref{th:impoinfinite}, by showing that the protocol $\mathcal P_{IP}$ cannot be simulated.
\end{proof}
One may wonder what would happen if we wanted to construct simulators that \vir{gracefully degrade} when omissions reach a certain threshold $t_{O}$. More precisely, for a sequence of interactions $I$ with $O(I) < t_{O}$, the simulator has to perform a full simulation of ${\cal P}$; if $O(I) \geq t_{O}$, the simulator has to start a simulation, but then it is allowed to stop forever in a ``consistent'' simulated state. Essentially, in the second case, we allow the sequence of events $E(\Gamma)$ defined in Section~\ref{s:simulation} to be finite (in other terms, we drop the simulator's ``liveness'' requirement).
 \begin{theorem} Given an infinite amount of memory on each agent, in the {\sf T$_3$} model (and hence in all the \emph{omissive} models of Figure~\ref{figure:modcompress}), any gracefully degrading  simulator that simulates all {\sf TW} protocols must have a threshold $t_O\leq 1$.
\end{theorem}
\begin{proof}
Recall that in Lemma~\ref{lemma:indstinguishability1} we constructed a sequence of interactions $I^*$ for a set of agents $A$, which was then applied to the protocol $\mathcal P_{IP}$ in order to prove Theorem~\ref{th:impoinfinite}. Suppose now that a simulator has threshold $t_O>1$. If such a simulator executes a run with at most one omission, it must effectively simulate infinitely many two-way interactions. In particular, it is able to simulate the first two-way interaction in a system of two agents, and therefore the sequence $I$ mentioned in Lemma~\ref{lemma:indstinguishability1} is well defined for this simulator, as well as the sequences $I_k$ and the numbers $t_k$. But then, as the agents of $A$ execute the same simulator according to the sequence $I^*$, they violate the safety property of {\sf Pair}, because $t+1$ of them change their simulated state from $c$ to $cs$. Since they reach a non-consistent simulated state, this means that a gracefully decaying simulator with threshold $t_O>1$ cannot simulate $\mathcal P_{IP}$.
\end{proof}

\begin{figure}
\begin{minipage}{0.525\textwidth}
\scriptsize
\begin{framed}
\begin{algorithmic}[1]

\State $my\_id=unique\_ID;$ $\, state_{\mathcal{P}}= initial\_state_{\mathcal{P}};$ $\, id_{other}=\bot;$ $\, state_{other}=\bot;$ $ \, state_{sim}=available$  \Comment{Agent's variables}
\medskip

\State \sf{Upon Event}  {Reactor delivers} {$\,(id^{s}, state^{s}_{\mathcal{P}}, id^{s}_{other}, state^{s}_{other}, state_{sim}^{s})$} 

    	    \If {$(state_{sim}=available \, \wedge \, state_{sim}^{s}=available)$ } \label{id2:pair1}
    	       \State $state_{sim}=pairing$
	       \State $id_{other}=id^{s};$ $\, state_{other}=state^{s}_{\mathcal{P}}$ \label{id2:pair2}
	   	
	      \ElsIf {$(state_{sim}=available \, \wedge \, state_{sim}^{s}=pairing \, \wedge \, id^{s}_{other}=my\_id  \, \wedge \, state^{s}_{other}=state_{\mathcal{P}})$ }\label{id2:iflock}
	      \State $state_{sim}=locked$ \label{id2:createlock}
	      \State $id_{other}=id^{s};$ $\, state_{other}=state^{s}_{\mathcal{P}}$
	      \State $state_{\mathcal{P}}= \delta_{{\cal P}}(state_{\mathcal{P}},state_{other})[0]$ \label{id2:chst2}
	      
	    \ElsIf {($state_{sim}=pairing \, \wedge \, id_{other}=id^{s} \, \wedge \, id^{s}_{other}=my\_id  \wedge state_{sim}^{s}=locked$ )} \label{id2:ifotherlocked}
	        \State $state_{sim}=available$
	        \State $id_{other}=state_{other}=\bot$
	      
	        \State $state_{\mathcal{P}}= \delta_{\cal P}(state^{s}_{\mathcal{P}},state_{\mathcal{P}})[1]$  \label{id2:ifotherend}
	          
           \ElsIf {$(id_{other}=id^{s}  \, \wedge \,  id^{s}_{other} \neq my\_id )$ } \label{id2:iflockevolution}
			 \State $state_{sim}=available$  \label{id2:unlock}
	        	        \State $id_{other}=state_{other}=\bot$   \label{id2:endlock}
		 \EndIf 
	     	 
\end{algorithmic}
\end{framed}
\vspace{-0.35cm}
\caption{ Simulation protocol $\mathcal{S}_{ID}$ \label{id2:algorithm1}}
\end{minipage}%
\begin{minipage}{0.47\textwidth}
\hspace{0.5cm}
\includegraphics[scale=0.43]{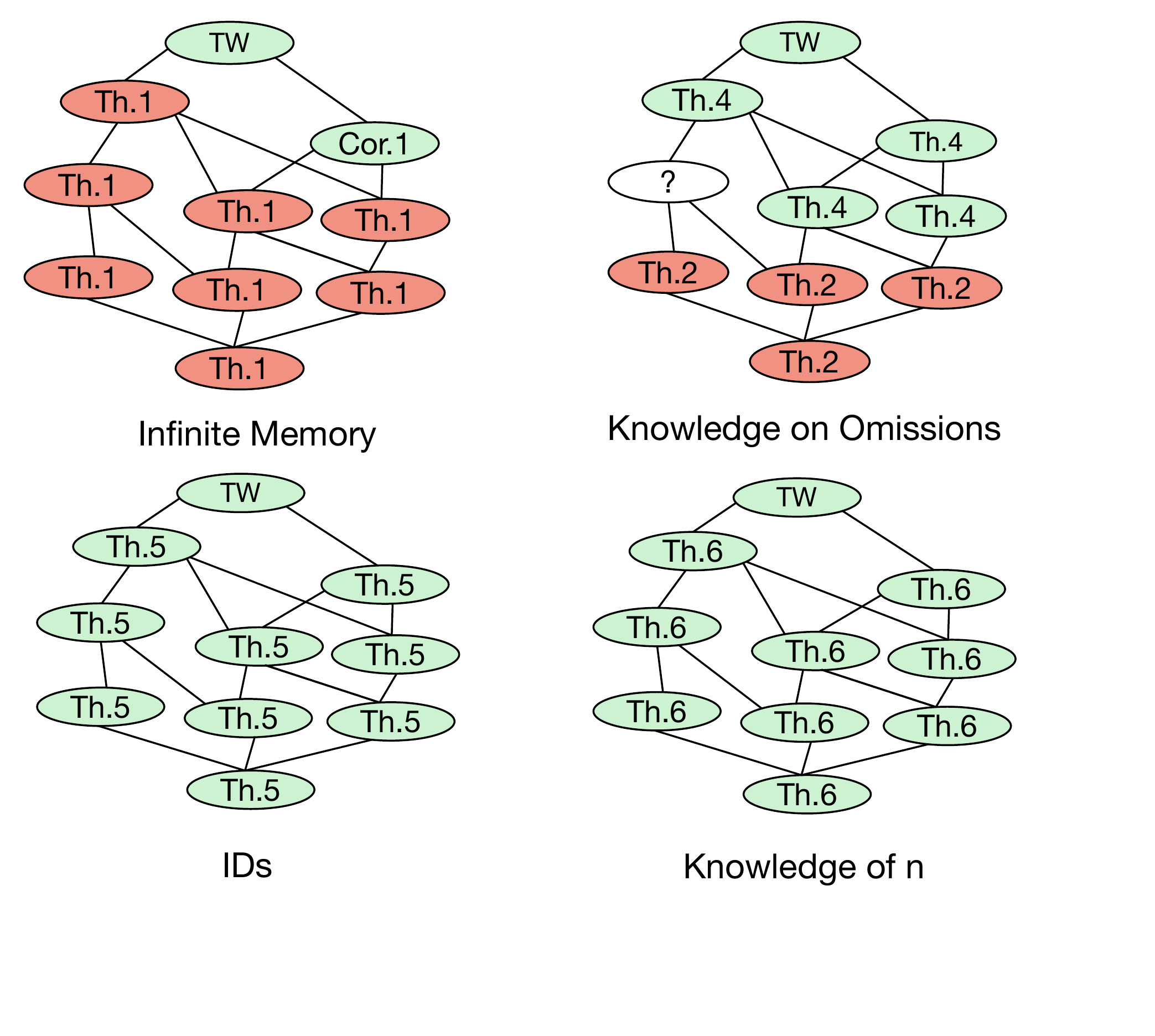}
\vspace{-1.65cm}
\caption{Map of results (cf.~Figure~\ref{figure:modcompress})\label{id:algorwq2}}
\end{minipage}
\end{figure}

\section{Simulation in Omissive Models} \label{sec:6}

In this section we focus on designing simulators of two-way protocols. In light of the impossibilities presented in the previous section, additional assumptions are necessary. Section~\ref{knowfaultsimulator} assumes some knowledge
on the maximum number of omissions, Section~\ref{idsimulator} assumes the presence of unique IDs, and finally in Section~\ref{nsimulator} we assume to know the number of agents. 

\subsection{Knowledge on Omissions: Simulator $\mathcal{S}_{KnO}$ \label{knowfaultsimulator}}

Here we assume to know an upper bound $o$ on the number of omissions, i.e., for any sequence of interactions $I$ on which the simulator runs we have $O(I) \leq o$. We will show that under this assumption there exists a simulator for models {\sf I}$_3$ and {\sf I}$_4$. 
This contrasts with models {\sf I}$_1$ and {\sf I}$_2$, in which it is impossible to simulate even when $O(I)\leq 1$ (see Theorem \ref{th:meminfimp1om}).

We explain the simulator $\mathcal{S}_{KnO}$ under model {\sf I}$_3$; the version for model {\sf I}$_4$ is only slightly different, and its correctness follows from symmetry considerations. 
The simulator is based on the exchange of ``tokens''. Each simulated state $q \in Q_{\cal P}$ is represented as a sequence of numbered tokens: $\langle q,1\rangle,\ldots, \langle q,o+1\rangle$. Intuitively, an agent tries to transmit its state to others by sending one token at a time, for $o+1$ consecutive interactions, each time incrementing the counter. When a reactor detects an omission, it generates a \emph{joker} token $\langle J\rangle$, which will also be sent in successive interactions. Note that there are never more than $o$ jokers circulating.
Every time an agent gets a new token, it checks if it owns the complete set of $o+1$ tokens representing some state $q$ and, if so, it simulates (part of) an interaction with a hypothetical partner in state $q$. If the complete set of tokens is not available, the agent is allowed to replace the missing tokens with the jokers that it currently owns. After the $o+1$ tokens have been used, they are discarded and withdrawn from circulation. However, if an agent uses some joker tokens, it ``takes note'' of what tokens these jokers are replacing. If later on the same agent obtains one of the tokens in this list, say $\langle q,i\rangle$, it turns $\langle q,i\rangle$ into a joker and removes $\langle q,i\rangle$ from the list. (This is reminiscent of the card game Rummy.)

\smallskip
\noindent{\bf Simulator Variables.}  
Each agent has a queue of tokens to be sent, called $sending$, initially empty. It also has a variable $state_{sim}={\sf available}$ for the state of the simulator protocol, a variable $state_{\mathcal{P}}$ for the state of the simulated protocol (initialized according to its initial simulated state), and a multi-set of tokens called $Jokers$, initially empty.

\smallskip
\noindent  {\bf Simulator Protocol.}
Suppose that an agent interacts as a starter. If $state_{sim}={\sf available}$ and $sending$ is empty, the agent switches to $state_{sim}={\sf pending}$ and inserts the complete set of tokens $\langle state_{\mathcal{P}},1\rangle,\ldots, \langle state_{\mathcal{P}},o+1\rangle$ into $sending$. In any case, and regardless of $state_{sim}$, the starter removes the first token from the queue, and the reactor reads it.

Suppose now that an agent interacts as a reactor. To begin with, it reads the first token from the $sending$ queue of the starter, and enqueues it into its own $sending$ queue. If it detects an omission, it enqueues a joker token instead. Then it performs a preliminary check: if $state_{sim}={\sf pending}$, and the agent can find a complete set of tokens for its own state (i.e., $state_{\mathcal{P}}$) in its own $sending$ queue (possibly using some joker tokens as wildcards), it switches to $state_{sim}={\sf available}$ and removes the set of $o+1$ used tokens from the queue. After this preliminary check, the core protocol starts: if $state_{sim}={\sf available}$, and the agent has a complete set of tokens for some state $q$ in its own $sending$ queue (possibly using some joker tokens), it removes the set of $o+1$ used tokens from the queue, it simulates its part of the two-way transition with an agent in state $q$ (i.e., it updates $state_{\mathcal{P}}=\delta(q,state_{\mathcal{P}})[1]$), and it enqueues into $sending$ a complete set of ``state change'' tokens, i.e., $\langle(q,state_{\mathcal{P}}),1\rangle,\ldots,\langle(q,state_{\mathcal{P}}),{o}+1\rangle$. On the other hand, if $state_{sim}={\sf pending}$, and the agent has a complete set of state change tokens of the form $\langle(state_{\mathcal{P}},q'),i\rangle$ in its own $sending$ queue (possibly using some joker tokens), it removes the set of $o+1$ used tokens from the queue, it updates $state_{\mathcal{P}}=\delta(state_{\mathcal{P}},q')[0]$, and switches to $state_{sim}={\sf available}$.

Also, whenever a reactor uses a joker token as a substitute for some token $\langle q,i\rangle$, it adds $\langle q,i\rangle$ to the multi-set $Jokers$. Symmetrically, when it receives a new token $\langle q,i\rangle$ from a starter and that token is in $Jokers$, it removes one copy of $\langle q,i\rangle$ from $Jokers$, removes the last copy of $\langle q,i\rangle$ from $senders$, and enqueues a new joker token into $senders$.


\begin{lemma}
The derived execution of simulator $\mathcal{S}_{KnO}$ is {\sf GF}.
\end{lemma}
\begin{proof}

 Let $\mathcal C$ and $\mathcal C'$ be closed sets of configurations of $\mathcal P$, such that every configuration of $\mathcal C$ can become one of $\mathcal C'$ after a two-way interaction, and suppose that the derived execution passes through $\mathcal C$ infinitely many times. Let $\widetilde{\mathcal C}$ be the set of configurations of the simulator protocol whose simulated states are in $\mathcal C$, and let $\widetilde{\mathcal C'}$ be constructed similarly from $\mathcal C'$. By assumption, the simulation passes through $\widetilde{\mathcal C}$ infinitely often; we claim that it must go through $\widetilde{\mathcal C'}$ infinitely many times, as well. By definition of $\mathcal C$, for every $C_j\in \widetilde{\mathcal C}$, there is an interaction in $\mathcal P$ between $a_s$ and $a_r$ that maps $\pi_\mathcal P(C_j)$ into $\pi_\mathcal P(C'_j)$, where $C'_j\in \widetilde{\mathcal C'}$ is obtained from $C_j$ by changing the simulated states of $a_s$ and $a_r$ according to $\delta_\mathcal P$ (and possibly some simulator variables). Since in $C_j$ there are no ``pending'' transactions, $C_j$ has a possible continuation $\widehat{C}_j$ where agents $a_s$ and $a_r$ are {\sf available}, such that $\pi_\mathcal P(\widehat{C}_j)=\pi_\mathcal P(C_j)$, and from this configuration the final configuration $\widetilde{\mathcal C'}$ is reachable. 
 The proof is done by case analysis on the number of agents:
 \begin{itemize}
 \item Case $n > 4$: In this case a simple counting arguments shows that it is always possible to find a finite sequence of interactions such that: agents $a_s$ and $a_r$ are available and have an empty buffer and no agent sets $state_{\mathcal{P}}=\delta(q_s,state_{\mathcal{P}})[1]$. Let $q_r,q_s$ be the state of $a_r,a_s$, respectively. Having at most $n(o+1)$ tokens, distributing this tokens among $n-2$ agents, such that each agent obtains at most $(1+\frac{2}{n-2})(o+1) < (1+2/3)(o+1)$ tokens. No agent can set $state_{\mathcal{P}}=\delta(q_s,state_{\mathcal{P}})[1]$ if it does not receives at least $2(o+1)$ tokens, please note that by assumption we started by a configuration where agents were either $pending$ or $available$ and no run $<(q_s,*),*>$ was present. 
 Starting from this configuration, we first show how to "unlock" agent $a_r$ (or $a_s$ if it is pending). If agent $a_r$ is pending there is an agent $a_x$ containing the first token of the run $<q_r,*>$, this token can be in position $j$ in the buffer of $a_x$ with $0<j<(1+2/3)(o+1)$. If $j=0$ the token is sent to $a_r$ by doing an interaction $(a_x,a_r)$ otherwise we do first a sequence of $j-1$ interactions $(a_x,a_r)$, note that this could at most trigger the exit of $a_r$ from the $pending$ state since $a_r$ receives strictly less then $2(o+1)$ tokens, if this is the case we have done. In any case we append another sequence of interactions $(a_r,a_x)$ until the buffer of $a_r$ is empty. Now either $a_r$ is available, or the token in position $j=0$ in the buffer of $a_x$ is the first token of the run of $a_r$, in this case $(a_x,a_s)$  moves this token to the buffer of $a_s$. This procedure is iterated until all tokens but one are in the buffer of $a_s$, the last token is in position $j=0$ of agent $a_y$ and the buffer of $a_r$ is empty, when this happens a trivial sequence of interactions $(a_s,a_r)$ followed by $(a_y,a_r)$ bring $a_r$ in state $available$ with an empty buffer and it empties the buffer of $a_s$.
 The same procedure can be used to unlock $a_s$ if it is pending. 
 Once $a_s$ and $a_r$ are both available and with an empty buffer a sequence of $(o+1)$ interactions $(a_s,a_r)$ followed by a sequence of $(o+1)$ interactions $(a_r,a_s)$ brings the system in configuration $\widetilde{\mathcal C'}$.
 
 \item Case $n=4,3$: Let us first examine the case of three agents, and let $a_x$ with state $q_x$ be the third agent. Let us assume that $a_r$ does not have in the buffer a complete run for $<q_x,*>$ otherwise it sends its tokens to $a_x$ until the first token of the run for $q_x$ is sent to $a_x$. 
 Then we send all tokens for the run $<q_r,*>$  to $a_r$: if agent $a_x$ has some token for $<q_s,*>$ or $<q_r,*>$ in position $0$ of the buffer we have an interaction $(a_x,a_r)$, the same is done for agent $a_s$ this procedure stop when both agents $a_s,a_x$ have a token for $<q_x,*>$ in position $0$. When this happen we let agents have interactions $(a_s,a_x)$ and $(a_s,a_x)$, this allows the agent to remove the tokens for $q_x$, note that during this procedure agent $a_s$ does not increases its number of tokens for $<q_x,*>$ contained in buffer. Iterating this procedure agent $a_r$ obtains a complete run $<q_r,*>$ and $<q_s,*>$. At this points it exits from pending being available and it executes the first portion of the simulated two interactions. A series of successive $o+1$ interactions $(a_r,a_s)$ brings the system in configuration $\widetilde{\mathcal C'}$.
 The case for $n=4$ agents is analogous. 
 
 \item Case $n=2$: In this case after at most $2(o+1)$ interactions $(a_s,a_r)$ and after other $(o+1)$ interactions $(a_r,a_s)$ we have that the agents changes simulated state bringing the system in configuration $\widetilde{\mathcal C'}$.
 \end{itemize}
\end{proof}

\begin{theorem}
Assuming {\sf I}$_{3}$ or {\sf I}$_4$ and ${\Theta}(\log n |Q_{\cal P}|(o+1) )$ bits of memory on each agent, there exists a protocol that simulates every {\sf TW} protocol.
\label{th2:simO}
\end{theorem}
  \begin{proof}
   The proof uses model {\sf I}$_3$, the correctness for model {\sf I}$_4$ follows from symmetry consideration.
  We first show that an agents sets $state_{\mathcal{P}}=\delta(q,state_{\mathcal{P}})[1]$  infinitely many times. Let us first consider the case where an agent $a_{r}$ with $state_{sim}={\sf available}$ exists, if there exists an agent $a_{s} \neq a_{r}$ with state $state_{sim}={\sf available}$ it is easy to see that after $o+1$ interactions  $(a_s,a_r)$ , despite the presence of omissions, agent $a_{r}$ will have a run for the state $q_s$ of $a_s$ in $sending$, therefore it executes  $state_{\mathcal{P}}=\delta(q_s,state_{\mathcal{P}})[1]$. In case $a_s$ is in ${\sf pending}$ and there is no token $<(q_s,*),*>$, we show that there exists a run for a state $a_s$ scattered among agents. 
  
We claim that once a run is created it disappears from the system only if it is consumed by an agent.  Suppose that there is no token $<q_s,*>$ this implies that each time an agent was trying to transmit a toke for state $q_s$ an omission occurred, but there are $o+1$ such tokens and at most $o$ omissions, therefore this is impossible. 
Now let us suppose that the number of jokers and tokens $<q_s,*>$ is less than $o+1$: each time an omission occurred when a token $<q_s,*>$ of a specific run was sent the receiver generated a joker, let $T$ be the set of these jokers, we have that $|T|$ and the number of tokens $<q_s,*>$ is at least $o+1$. Note that if one token in $T$ was used by an agent then there exists an agent $a_1$ with a token $<q_x,l> \in Jokers$ and either: (1) there exists a token $<q_x,l>$ on agent $a_2$, therefore after a finite number of interactions $(a_2,a_1)$ we have that agent $a_1$ will put a joker in $sending$; or (2) if token $<q_x,l>$ does not exists then it was lost during an interaction but a corresponding joker was generated. 

Therefore the run for state $q_s$ exists and thus it exists a sequence of interactions that move this run, or another, in the $sending$ buffer of $q_r$, therefore $a_{r}$ can execute $state_{\mathcal{P}}=\delta(q_s,state_{\mathcal{P}})[1]$.

If there are only tokens $<(q_s,*),*>$ then by using the previous reasoning we can show that there is a run of such tokens, from this point on we boil down to a previous case. 
A similar argument shows that  if there are both tokens,  $<q_s,*>$ and $<(q_s,*),*>$ there is a run for at least one of them. 

The only case left is if there is no agent with $state_{sim}={\sf available}$, in this case we show that it must exists an agent $a_{r}$ with $state_{{\cal P}}=q_r$ and a run of tokens $<q_r,*>$ or $<(q_r,*),*>$.
Let us assume the contrary, when an agent switches to state ${\sf pending}$ it inserts a run of tokens $<q_r,*>$  in $sending$. This run can only disappears if is consumed by an agent, but in this case a run $<(q_r,*),*>$ is created,
now if this run is consumed by another agent $a_x$ with state $q_r$ this implies that the run $<q_r,*>$ of $a_x$ is still in the system. 

Therefore there always exists a sequence of interactions that bring $a_r$ to state ${\sf available}$.

Now we have to show that each time an agent $a_r$ sets $state_{\mathcal{P}}=\delta(q_s,state_{\mathcal{P}})[1]$, it can be paired consistently in the matching. When, $a_r$ changes state, at time $t$ it consumes a run $<q_s,*>$,
this implies that there exists an instant $t' < t$ where an agent $a_s$ has generated the run $<q_s,*>$ entering in state $pending$. Now $a_s$ could exit from $pending$ only if it consumes a run  $<q_s,*>, <(q_s,*),*>$ if is the run generated by $a_r$ we have the edge of the matching, otherwise this implies that there exist another agent $a_b$ with state $q_s$ that generated a run $<q_s,*>$ and consumed the run of $ <(q_s,*),*>$ of $a_r$ at time $t''$, for the moment let us assume $b \neq r$, if $a_s$ was pending at time $t''$, then, being agent anonymous, we can switch the role of $a_s$ and $a_r$ and match the two agents. Otherwise if $a_s$ was not pending, we have that when $a_s$ exited from state ${\sf pending}$ there must exists another agent $a_x$ with state $q_s$ in pending, therefore we can switch the role of $a_s$ and $a_x$. Let us now study the case where $b=r$, in this case we have $\delta(q_s,state_{\mathcal{P}})[1]=q_s$ thus when $a_r$ consumes the run $<(q_s,q_r),*>$ generated by itself it enters in state $\delta(q_s,state_{\mathcal{P}})[0]=q_x$, therefore we can match $a_s,a_r$, switching the role of $a_s$ and $a_r$, since it is as they transitioned from $(q_s,q_r)$ to  $(q_x,q_s)$.
So we have shown that under {\sf GF} the matching is not empty and contains infinitely many pairs.

Finally, the derived execution is {\sf GF} due to Lemma~2.
\end{proof}

  By applying this theorem to a system without omissions (i.e., plugging $o=0$), we have:
\setcounter{theorem}{0}
\begin{corollary}Given ${\Theta}(|Q_{\cal P}|\log n)$ bits of memory on each agent, every {\sf TW} protocol can be simulated in {\sf IT}.\qed\label{obs:1}
\end{corollary}
\setcounter{theorem}{4}

\subsection{Unique IDs and {\sf IO}: Simulator $\mathcal{S}_{ID}$. \label{idsimulator}}
Now we assume that the agents have unique IDs as part of their initial state, and we give a {\sf TW} simulator for the {\sf IO} model, named $\mathcal{S}_{ID}$, which is reported in Figure~\ref{id2:algorithm1}. The idea is to use the uniqueness of the IDs to implement a locking mechanism that ensures the consistent matching of simulated state changes. Essentially, at a certain point an agent commits itself to executing a transition only with another agent with a specific ID. The locking scheme contains a rollback procedure to avoid deadlocks. 
 \smallskip

  \noindent{\bf Simulator Variables.}  
  Each agent has the following variables: $my\_id$ for its own ID, $state_{sim}={\sf avaible}$ for the state of the simulator protocol, and $state_{\mathcal{P}}$ for the state of the simulated protocol. Moreover, it keeps two variables, $id_{other}$ and $state_{other}$, which are the ID and the state of the other agent in the simulated two-way interaction. 

\smallskip

  \noindent  {\bf Simulator Protocol.} When an {\sf available} reactor $a_{r}$, with ID ${r}$ and  simulated state  $state_{\mathcal{P}}=q_r$, observes a starter $a_{s}$ with ID $s$ and $state^{s}_{sim}={\sf available}$, it enters a {\sf pairing} state. Moreover, it saves the ID $s$  in $id_{other}$ and the simulated state $state^{s}_{\mathcal{P}}=q_s$ of $a_{s}$ in $state_{other}$ (see the details at Lines~\ref{id2:pair1}--\ref{id2:pair2}). The {\sf pairing} state could be seen as a \vir{soft} commitment in which a reactor picks a specific agent as a possible partner for a two-way interaction. In some specific conditions, an agent in the {\sf pairing} state can \vir{roll back} to  the {\sf available} state without completing a simulated two-way interaction; this will be covered later. 
   
The simulation proceeds as soon as $a_{s}$, which is {\sf available}, receives the information that some other agent $a_{r}$ is in the {\sf pairing} state and  wants to pair up with an agent that has $my\_id=s$ and  simulated state $q_s$. In this case $a_{s}$ sets its simulator state to {\sf locked}, stores $a_r$'s simulated state and ID, and executes the transition $\delta_{\cal P}(state_{\mathcal{P}},state_{other}=q_r)[0]=f_s(state_{\mathcal{P}},state_{other})$. We remark that this happens only if the current simulated state of $a_s$ is equal to the variable $state_{other}$ of $a_r$ (see Line~\ref{id2:iflock}).

Suppose that $a_{s}$ is {\sf locked}; if $a_{r}$ observes $a_{s}$,  it executes the transition $\delta_{\cal P}(state^{s}_{\mathcal{P}},state_{\mathcal{P}})[1]=f_r(state^{s}_{\mathcal{P}},state_{\mathcal{P}})$, becomes {\sf available}, and resets the variable $id_{other}$ (see Lines~\ref{id2:ifotherlocked}--\ref{id2:ifotherend}). Now, if $a_{s}$ is {\sf locked} and observes that $a_{r}$'s variable $id_{other}$ is not $s$, then it resets its own state to {\sf available} (see Lines~\ref{id2:iflockevolution}--\ref{id2:endlock}).

  It may happen, due to the {\sf IO} model's nature, that $a_{s}$, with variable $state_{\mathcal{P}}=q_s$, induces an agent $a_{r}$ to enter state {\sf pairing}, but then $a_{s}$ starts a two-way simulation with a different agent. In order to prevent $a_r$ from waiting forever, we make it reset the pending transition if it encounters $a_{s}$ again with  $id_{other}^{s} \neq my\_id$ (this is incorporated in Lines~\ref{id2:iflockevolution}--\ref{id2:endlock}).

\begin{theorem}
Assuming {\sf IO} and unique IDs,
$\mathcal{S}_{ID}$  is a {\sf TW} simulator.
\label{th2:simid1}
\end{theorem}
\begin{proof}
Let us consider the simulation of a generic two-way protocol ${\cal P}$. Assume that an agent $a_0$ becomes {\sf pairing} upon observing an agent $a_1$. Later, $a_1$ can either become {\sf locked} with $a_0$, or {\sf pairing} as well, upon observing some other agent $a_2$. It is clear that, if such a ``chain'' of {\sf pairing} agents is formed, it must stop eventually. The last agent in the chain, say $a_k$, will then have to become {\sf locked} upon observing some {\sf pairing} agent with $id_{other}$ equal to $a_k$'s ID (which will eventually happen due to the {\sf GF} condition).

Now, whenever an agent $a_s$ enters state {\sf locked} after observing an agent $a_r$ in state {\sf pairing}, it changes its simulated state according to $\delta_\mathcal P$, say at time $t_s$, and sooner or later also $a_r$ will do the same, say at time $t_r$, with $t_r>t_s$. This is because $a_r$ cannot start a new interaction with $a_s$ between times $t_s$ and $t_r$ (since $a_s$ would have to be in state {\sf available}), and hence it will necessarily be seen by $a_s$ with $id_{other}\neq s$, due to the {\sf GF} condition. Moreover, $a_s$ cannot change its own simulated state after $t_s$ and before $t_r$, because it is {\sf locked}.

We have proved that infinitely many simulated state transitions must occur; these events can easily be paired up into a consistent perfect matching. We only have to prove that the derived execution satisfies the {\sf GF} condition. We will do it in the case in which the system consists of $n\geq 3$ agents; the proof for the case $n=2$ is simpler, and we omit it. Let $\mathcal C$ and $\mathcal C'$ be closed sets of configurations of $\mathcal P$, such that every configuration of $\mathcal C$ can become one of $\mathcal C'$ after a two-way interaction, and suppose that the derived execution passes through $\mathcal C$ infinitely many times. Let $\widetilde{\mathcal C}$ be the set of configurations of the simulator protocol whose simulated states are in $\mathcal C$ and let $\widetilde{\mathcal C'}$ be constructed similarly from $\mathcal C'$. (Note: if a configuration of the simulator protocol contains a {\sf locked} agent $a_s$, the simulated state of its partner $a_r$ is assumed to be the state it would reach after the interaction with $a_s$. This agrees with the definition of derived run given in Section~\ref{s:simulation}.) By assumption, the simulation passes through $\widetilde{\mathcal C}$ infinitely often; we claim that it must go through $\widetilde{\mathcal C'}$ infinitely many times, as well. By definition of $\mathcal C$, for every $C_j\in \widetilde{\mathcal C}$, there is an interaction in $\mathcal P$ between two agents $a_s$ and $a_r$ that maps $\pi_\mathcal P(C_j)$ into $\pi_\mathcal P(C'_j)$, where $C'_j\in \widetilde{\mathcal C'}$. We will prove that such a $C'_j$ can be reached from $C_j$ after at most a constant number of interactions. 
\begin{itemize}
\item If $a_s$ is {\sf available} in $C_j$ and $a_r$ is either {\sf available} or {\sf pairing} with $a_s$, then $C'_j$ can be obtained by simply letting $a_s$ and $a_r$ interact together multiple times until they perform a full simulated interaction, and their states transition according to $\delta_\mathcal P$.
\item If $a_s$ or $a_r$ (perhaps both) is {\sf locked} in $C_j$, we let it interact with its current partner until the simulated interaction is completed and its internal state is again {\sf available}. Then we proceed as in the other cases.
\item If $a_s$ is {\sf pairing} in $C_j$ or $a_r$ is {\sf pairing} with an agent that is not $a_s$, we have to make it become {\sf available} without performing a full two-way interaction, and then we can proceed as in the other cases. Suppose that $a_s$ is {\sf pairing} (the case with $a_r$ is handled similarly), and let $a_q$ be the agent with which $a_s$ is paired (perhaps $a_q=a_r$).
\begin{itemize}
\item If $a_q$ is {\sf pairing} in $C_j$ (of course not with $a_s$), then we let $a_s$ observe $a_q$ and roll back to the {\sf available} state.
\item If $a_q$ is {\sf available} in $C_j$, we let it pair up with some other {\sf available} agent (possibly $a_r$), and the we proceed as in the previous case. If an {\sf available} agent does not exist, we can create one by letting some {\sf pairing} agent roll back or some {\sf locked} agent complete its current interaction, as explained in the first paragraph of the proof.
\item If $a_q$ is {\sf locked} in $C_j$, we let it finish the simulated interaction and become {\sf available}. If it was {\sf locked} with $a_s$, we are finished because now $a_s$ is {\sf available} too. Otherwise, we proceed as in the previous case.
\end{itemize}
\end{itemize}
As already observed, $C'_j$ can be reached from $C_j$ after at most a constant number $c$ of interactions, and this holds for every $j$. By applying the definition of {\sf GF} to the simulator's execution $c$ times, we have that $C'_j$ is indeed reached for infinitely many $j$'s. Therefore $\widetilde{\mathcal C'}$ is reached infinitely many times, and thus so is $\mathcal C'$ by the derived execution.
\end{proof}
\subsection{Simulating with knowledge of $n$ \label{nsimulator}}

We give the following additional result on simulating when additional knowledge is available to the agents. The protocol uses a naming algorithm  ${\cal N}_{n}$ in conjunction with ${\cal S}_{ID}$. 

\subsubsection*{ Naming Algorithm: ${\cal N}_{n}$} The following naming protocol, ${\cal N}_{n}$,  uses the knowledge of $n$. This naming protocol is similar to the threshold protocol for {\sf IO} presented in \cite{oneway}.
 \\ \noindent{\bf Protocol Variables.}  
Each agent $a_r$ has variables $my\_id=1$,$max\_id=1$, moreover it can access simulator ${\cal S}_{ID}$ by function ${\sf start\_sim}(id)$ that takes as input an unique id.
 \\ \noindent{\bf Simulator Protocol.}  
 When an agent $a_r$ is the responder of an interaction and the starter has its same value for $my\_id$ it increments the value of its variable $my\_id$. Similarly, variable $max\_id$ is updated to reflect the maximum value seen for variable $my\_id$, $a_r$ updates the value $max\_id$ to the maximum value between the values $my\_id,max\_id$ of the initiator and its variables. When value $max\_id=n$  the agent invokes ${\sf start\_sim}(max\_id)$.

\begin{lemma}
Let us consider a set of agents $A$ running ${\cal N}_{\diamond n}$ algorithm. The system fairness is {\sf GF} and the interaction model is {\sf IO}. The maximum value $M$ for $max\_id$ increases if and only if there are two agents that share the same value for $my\_id$. When $M=|A|$ each agent has an unique stable value for $my\_id$. 
\label{lemma:eventnaming}
\end{lemma}
\begin{proof}

If the maximum value $M$ for $max\_id$ increased there exists an agent $a$ that at the end of an interaction has $my\_id > max\_id=M$, but $a$ increases $my\_id$ only when it is the responder of an interaction with an agent $a'$ with same value for $my\_id=M$. For the other direction let us suppose that there are two agents $a,a'$ with same value for $my\_id$, if one of them has value $my\_id=M$ then for {\sf GF} they will eventually interacts and this increase $M$, if they have a value $my\_id=M'<M$ they will eventually interacts and one of them, let us suppose $a'$, will increase its value for $my\_id=M''=M'+1$. By definition of $M$ there must exist another agent $a''$ with $my\_id=M''$, thus we can iterate this reasoning until we have two agents with equal $my\_id=M$.
It remains to show that eventually all agents will assume an unique value for $my\_id$ in $[1,\ldots,|A|]$, but this is trivial by observing that given a set of agents with equal value for $my\_id=M'$ there will be interactions by them until only two agents with value $my\_id=M'$ remains. When this happen, they eventually interact and only one agent with $my\_id=M'$ remains, this agent never changes $my\_id$, the rest derives from the initial state $my\_id=1$ for all agents. 
\end{proof}

The correctness of this protocol derives immediately from Lemma~\ref{lemma:eventnaming} and Theorem~\ref{th2:simid1}. Thus, the following Theorem holds:

\setcounter{theorem}{5}
\begin{theorem}\label{th2:simN}
Assuming {\sf IO}, knowledge of $|A|=n$, and ${\Theta}(\log n)$ bits of memory, there exists a simulator for every {\sf TW} protocol.
\end{theorem}

\section{Conclusion}
In this paper we have given a formal definition of two-way simulation in Population Protocols, and we identified several omission models.
On top of this framework, we have given several impossibility results, as well as two-way simulators.
Our results yield an almost comprehensive characterization, see Figure~\ref{id:algorwq2}.
The only gap left concerns the possibility of simulation in model {\sf T}$_2$ when an 
upper bound on the number of omissions is known. 
As future work we are going to investigate this gap and study models where a unique
leader agent is present. Our preliminary results in the latter direction show that the problem is far from trivial, and two-way simulation is still impossible in a wide set of models.

\bibliographystyle{plain}

\newpage

\end{document}